\documentclass[letterpaper, 10 pt, conference]{ieeeconf}  

\IEEEoverridecommandlockouts                              

\usepackage{graphics} 
\usepackage{epsfig} 
\usepackage{amsmath} 
\DeclareMathOperator*{\argmax}{arg\,max}

\usepackage{amssymb}  
\usepackage{graphicx}
\usepackage{url}

\usepackage{xcolor}

\newtheorem{theorem}{Theorem}

\newtheorem{lemma}{Lemma}
\newtheorem{remark}{Remark}

\title{\LARGE \bf
A Risk-Sensitive Finite-Time Reachability Approach for \\Safety of Stochastic Dynamic Systems}

\author{Margaret P. Chapman$^{1,2}$, Jonathan Lacotte$^{3}$, Aviv Tamar$^{1}$, Donggun Lee$^{4}$,
Kevin M. Smith$^{5}$,\\ Victoria Cheng$^{6}$, Jaime F. Fisac$^{1}$, Susmit Jha$^{2}$, Marco Pavone\textsuperscript{7}, Claire J. Tomlin$^{1}$
\thanks{$^{1}$M.C., J.F., A.T., and C.T. are with the Department of Electrical Engineering and Computer Sciences, University of California, Berkeley, USA.
        {\tt\small chapmanm@berkeley.edu}}%
\thanks{$^{2}$S.J. is with Computer Science Laboratory at SRI International, Menlo Park, California, USA. M.C. was a Student Associate at SRI International.
        }%
\thanks{$^{3}$J.L. is with the Department of Electrical Engineering, Stanford University, USA.
        }%
\thanks{$^{4}$D.L. is with the Department of Mechanical Engineering, University of California, Berkeley, USA.
        }%
\thanks{$^{5}$K.S. is with OptiRTC, Inc. and the Department of Civil and Environmental Engineering, Tufts University, USA.
        }%
\thanks{$^{6}$V.C. is with the Department of Civil and Environmental Engineering, University of California, Berkeley, USA.
        }%
        \thanks{$^{7}$M.P. is with the Department of Aeronautics and Astronautics, Stanford University, USA.
        }%
}

\newif\ifmargincomments 
\margincommentstrue

\ifmargincomments 

\else

\fi

\begin{document}

\maketitle

\thispagestyle{empty}
\pagestyle{empty}

\begin{abstract}
A classic reachability problem for safety of dynamic systems is to compute the set of initial states from which 
the state trajectory is guaranteed to stay inside a given constraint set over a given time horizon. 
In this paper, we leverage existing theory of reachability analysis and risk measures to devise a \textit{risk-sensitive} reachability approach for safety of {\em stochastic} dynamic systems under non-adversarial disturbances
over a finite time horizon. Specifically, we first introduce the notion of a \textit{risk-sensitive safe set} as a set of initial states from which 
the risk of large constraint violations can be reduced to a required level via a control policy, where risk is quantified 
using the \textit{Conditional Value-at-Risk} (CVaR) measure. Second, we show how the computation of a risk-sensitive safe set can be reduced to the 
solution to a Markov Decision Process (MDP), where cost is assessed according to CVaR. 
Third, leveraging this reduction, we devise a tractable algorithm to approximate a risk-sensitive safe set 
and provide arguments about its correctness. 
Finally, we present a realistic
example inspired from stormwater catchment design to demonstrate the utility of risk-sensitive reachability analysis.
In particular, our approach allows a practitioner to tune the level of risk sensitivity from worst-case 
(which is typical for Hamilton-Jacobi reachability analysis) to risk-neutral 
(which is the case for stochastic reachability analysis). 
\end{abstract}


\section{Introduction}
\label{sec:introduction}
Reachability analysis is a formal verification method based on optimal control theory that is used to prove 
safety or performance properties of dynamic systems~\cite{bansal2017hamilton}.
A classic reachability problem for safety is to compute the set of initial states from which 
the state trajectory is guaranteed to stay inside a given constraint set over a given time horizon.
This problem was first considered for discrete-time dynamic systems by Bertsekas and Rhodes 
under the assumption that disturbances are uncertain but belong to known sets~\cite{bertsekas1971control},~\cite{bertsekas1971minimax},~\cite{bertsekas2005dynamic}.
In this context, the problem is solved using a minimax formulation,
in which disturbances behave adversarially and safety is described as a binary notion based on set membership~\cite{bertsekas1971control},~\cite{bertsekas1971minimax},~\cite[Sec. 3.6.2]{bertsekas2005dynamic}.

In practice, minimax formulations can yield overly conservative solutions, particularly because disturbances are usually non-adversarial.
Most storms do not cause major floods, and most vehicles are not involved in pursuit-evasion games.
If there are enough observations of the system, one can estimate a probability distribution
for the disturbance (e.g., see~\cite{silverman2018density}), and then assess safety properties of the system in a more realistic context.
For stochastic discrete-time dynamic systems, 
Abate et al.~\cite{abate2008probabilistic} developed an algorithm to compute a set of initial states
from which the probability of safety of the state trajectory can be increased to a required level by a control policy.\footnotemark
\footnotetext{Safety of the state trajectory is the event that the state trajectory stays in the constraint set over a finite time horizon.} 
Summers and Lygeros~\cite{summers2010verification} extended the algorithm of Abate et al. to quantify the probability of safety and performance
of the state trajectory, by specifying that the state trajectory should also reach a target set.   

Both the stochastic reachability methods~\cite{abate2008probabilistic},~\cite{summers2010verification} and the minimax reachability methods~\cite{bertsekas1971control},~\cite{bertsekas1971minimax},~\cite{bertsekas2005dynamic} for discrete-time dynamic systems
describe safety as a binary notion based on set membership.
In Abate et al., for example, the probability of safety to be optimized is formulated as an expectation of a product (or maximum)
of indicator functions, where each indicator encodes the event that the state at a particular time point is inside a given set~\cite{abate2008probabilistic}.
The stochastic reachability methods~\cite{abate2008probabilistic},~\cite{summers2010verification} 
do not generalize to quantify the distance between the state trajectory and the boundary of the constraint set,
since they use indicator functions to convert probabilities to expectations to be optimized.

In contrast, Hamilton-Jacobi (HJ) reachability methods quantify the deterministic analogue of this distance
for continuous-time systems subject to adversarial disturbances 
(e.g., see~\cite{bansal2017hamilton},~\cite{herbert2017fastrack},~\cite{EECS-2018-41},~\cite{mitchell2005toolbox}).
Quantifying the distance between the state trajectory and the boundary of the constraint set in a non-binary fashion
may be important in applications where the boundary is not known exactly,
or where mild constraint violations are inevitable, but extreme constraint violations must be avoided.

It is imperative that reachability methods for safety take into account the possibility that rare events can occur
with potentially damaging consequences. 
Reachability methods that assume adversarial disturbances (e.g.,~\cite{bansal2017hamilton},~\cite{bertsekas1971minimax}) suppose that harmful events will always occur,
which may yield solutions with limited practical utility, especially in applications with large uncertainty sets.
Stochastic reachability methods~\cite{abate2008probabilistic},~\cite{summers2010verification} do not explicitly account for rare high-consequence events,
as costs are evaluated in terms of an expectation.

In contrast, in this paper, we harness \textit{risk measure theory} to formulate a reachability analysis approach that explicitly accounts for the possibility of rare events with negative consequences: 
harmful events are likely to occur at some time, but they are unlikely to occur all the time.
Specifically, a \textit{risk measure} is a function that maps a random variable $Z$ representing a loss, or a cost, into the real line,
according to the possibility of danger associated with $Z$~\cite[Sec. 6.3]{shapiro2009lectures},~\cite[Sec. 2.2]{kisiala2015conditional}.
Risk-sensitive optimization
has been studied in applied mathematics~\cite{ruszczynski2010risk}, reinforcement learning~\cite{osogami2012robustness},~\cite{chow2015risk},~\cite{ratliff2017risk}, and optimal control~\cite{chow2014framework},~\cite{samuelson2018safety}.
A risk-sensitive method may provide more practical and protective decision-making machinery (versus stochastic or minimax methods) 
by encoding a flexible degree of conservativeness.

In this paper, we use a particular risk measure, called \textit{Conditional Value-at-Risk} (CVaR).
If $Z$ is a random variable representing cost with finite expectation, then the Conditional Value-at-Risk of $Z$ at the confidence level $\alpha \in (0,1]$
is defined as~\cite[Equation 6.22]{shapiro2009lectures},\footnotemark
\footnotetext{Conditional Value-at-Risk is also called \textit{Average Value-at-Risk}, which is abbreviated as AV@R in~\cite{shapiro2009lectures}.} 
\begin{equation}
\text{CVaR}_\alpha[Z] := {\underset{t \in \mathbb{R}}\min} \text{ }\Big\{ t + \frac{1}{\alpha}\mathbb{E}\big[\max\{Z-t,0\}\big] \Big\}.
\label{cvareqn}
\end{equation}
CVaR captures a full spectrum of risk assessments from risk-neutral to worst-case, since $\text{CVaR}_\alpha[Z]$ increases from $\mathbb{E}[Z]$ to $\text{ess}\sup Z$, as $\alpha$ decreases from 1 to 0.
CVaR has desirable mathematical properties for optimization~\cite{rockafellar2000optimization} and chance-constrained stochastic control~\cite{vitus2016stochastic}.
There is a well-established relationship between CVaR and chance constraints
that we will use to obtain probabilistic safety guarantees in this paper.
Please see~\cite{kisiala2015conditional} and~\cite{serraino2013conditional} for additional background on CVaR.

{\em Statement of Contributions.} This paper introduces a {\em risk-sensitive} reachability approach for safety of stochastic dynamic systems under non-adversarial disturbances
over a finite time horizon. Specifically, the contributions are four-fold. 
First, we introduce the notion of a \textit{risk-sensitive safe set} as a set of initial states 
from which the risk of large constraint violations can be reduced to a required level via a control policy, 
where risk is quantified using the \textit{Conditional Value-at-Risk} (CVaR) measure. 
Our formulation explicitly assesses the distance between the boundary of the constraint set and the
state trajectory of a stochastic dynamic system. This is an extension of stochastic 
reachability methods (e.g.,~\cite{abate2008probabilistic},~\cite{summers2010verification}), which replace this distance with a binary random variable.
Further, in contrast to stochastic reachability methods, our formulation explicitly accounts for rare high-consequence events, by posing the optimal control problem
in terms of CVaR, instead of a risk-neutral expectation. 
Second, we show how the computation of a risk-sensitive safe set can be reduced to the solution to a Markov Decision Process (MDP), 
where cost is assessed according to CVaR. Third, leveraging this reduction, we devise a tractable algorithm to approximate a risk-sensitive safe set 
and provide theoretical arguments to justify its correctness. 
Finally, we present a realistic
example inspired from stormwater catchment design to demonstrate the utility of risk-sensitive reachability analysis.

{\em Organization.} The rest of this paper is organized as follows. We present the problem formulation and define risk-sensitive safe sets in Sec.~\ref{sec::problem}. 
In Sec.~\ref{sec::reduction}, we show how the computation of a risk-sensitive safe set can be reduced to the solution to a CVaR-MDP problem, 
i.e., an MDP where cost is assessed according to CVaR. 
In Sec.~\ref{sec::alg}, we present a value-iteration algorithm to approximate risk-sensitive safe sets, along with theoretical arguments that support its correctness. 
In Sec.~\ref{sec::ex}, we provide numerical experiments on a realistic example inspired from stormwater catchment design. 
Finally, in Sec.~\ref{conc}, we draw conclusions and discuss directions for future work.
%
\section{Problem Formulation}\label{sec::problem}
\subsection{System Model}
We consider a fully observable stochastic discrete-time dynamic system over a finite time horizon~\cite[Sec. 1.2]{bertsekas2005dynamic},
\begin{equation}
x_{k+1} = f(x_k,u_k,w_k), \quad k = 0, 1, \dots, N-1,
\label{sys}
\end{equation}
such that $x_k \in \mathcal{X} \subseteq \mathbb{R}^n$ is the state of the system at time $k$,
$u_k \in U$ is the control at time $k$, and
$w_k \in D$ is the random disturbance at time $k$. The control space $U$ and disturbance space $D$ are finite sets of real-valued vectors.
The function $f : \mathcal{X} \times U \times D \rightarrow \mathcal{X}$ is bounded and Lipschitz continuous.
The probability that the disturbance equals $d_j \in D$ at time $k$ is $\mathbb{P}[w_k = d_j] = p_j$, 
where $0 \leq p_j \leq 1$ and $\sum_{j=1}^W p_j = 1$. We assume that $w_k$ is independent of $x_k$, $u_k$, and disturbances at any other times.  
The only source of randomness in the system is the disturbance.
In particular, the initial state $x_0$ is not random. 
The set of \textit{admissible, deterministic, history-dependent control policies} is,
\begin{equation}
\Pi := \big\{ (\mu_0, \mu_1, \dots, \mu_{N-1}) \mid \mu_k: H_k \rightarrow U \big\},
\label{pi}
\end{equation}
where $H_k := \underbrace{\mathcal{X} \times \hdots \times \mathcal{X}}_{\text{(k+1) times}}$ is the set of state histories up to time $k$.
We are given a constraint set $\mathcal{K} \subseteq \mathcal{X}$, and the \textit{safety criterion} that 
the state of the system should stay inside $\mathcal{K}$ over time. 
For example, if the system is a pond in a stormwater catchment, then $x_k$ may be the water level of the pond in feet at time $k$,
and $\mathcal{K} = [0, 5)$ indicates that the pond overflows if the water level exceeds 5 feet.
We quantify the extent of constraint violation/satisfaction using a surface function that characterizes the constraint set.
Specifically, similar to \cite[Eq. 2.3]{EECS-2018-41}, let $g: \mathcal{X} \rightarrow \mathbb{R}$ satisfy,
\begin{equation}
x \in \mathcal{K} \iff g(x) < 0.
\label{g}
\end{equation}
For example, we may choose $g(x) = x - 5$ to characterize $\mathcal{K} = [0, 5)$ on the state space
$\mathcal{X} = [0, \infty)$.

\subsection{Risk-Sensitive Safe Sets}
A \textit{risk-sensitive safe set} is a set of initial states from which 
the risk of large constraint violations can be reduced to a required level via a control policy, where risk is quantified using the CVaR measure.
We use the term \textit{risk level} to mean the allowable level of risk of constraint violations.
Formally, the risk-sensitive safe set at the confidence level $\alpha \in (0,1]$ and the risk level $r \in \mathbb{R}$ is defined as,
\begin{subequations}
	\label{myS}
\begin{equation}
\mathcal{S}_\alpha^r := \{x \in \mathcal{X} \mid W_0^*(x,\alpha) \leq r \}, \\
\end{equation}
where 
\begin{equation}
\label{eq::rsobjective}
\begin{aligned}
& W_0^*(x,\alpha) := {\underset{\pi \in \Pi}\min} \text{ CVaR}_\alpha\big[ Z^\pi_{x} \big], \\
& Z^\pi_{x} := {\underset{k = 0, \dots, N}\max} \big\{ g(x_k) \big\},
\end{aligned}
\end{equation}
\end{subequations}
such that the state trajectory, $(x_0, x_1, ..., x_N)$, evolves according to the dynamics model~\eqref{sys} with the initial state $x_0 = x$ under the policy $\pi \in \Pi$. 
The surface function $g$ characterizes distance to the constraint set $\mathcal{K}$ according to~\eqref{g}. Note that the minimum in the definition of $W_0^*(x,\alpha)$ is attained, as the next lemma states.
\begin{lemma}[Existence of a minimizer]
	\label{lemma::infeqmin}
	For any initial state $x_0 \in \mathcal{X}$ and any confidence level $\alpha \in (0,1]$, there exists a policy $\pi^* \in \Pi$ such that 
	\begin{equation*}
	\begin{split}
	\text{ CVaR}_\alpha\big[ Z^{\pi^*}_{x} \big] = {\underset{\pi \in \Pi}\inf} \text{ CVaR}_\alpha\big[ Z^\pi_{x} \big] ={\underset{\pi \in \Pi}\min} \text{ CVaR}_\alpha\big[ Z^\pi_{x} \big].
	\end{split}
	\end{equation*}
\end{lemma}
\begin{proof}
Fix the initial state $x_0$. Since the control and disturbance spaces are finite, the set of states that could be visited (starting from $x_0$) is finite. 
Therefore, the set of policies restricted to realizable histories from $x_0$ is finite. Hence, the infimum must be attained by some policy $\pi^*$.
\end{proof}

In the next sections, we will present a tractable algorithm to approximately compute risk-sensitive safe sets at different levels of confidence and risk.
%
%
\subsection{Discussion}
Computing risk-sensitive safe sets, as defined by~\eqref{myS}, is well-motivated for several reasons.
Our formulation incorporates different confidence levels and non-binary distance to the constraint set. 
In contrast, the stochastic reachability problem addressed by Abate et al.~\cite{abate2008probabilistic} 
uses a single confidence level and an indicator function to measure distance to the constraint set,
in order to quantify the probability of constraint violation.
Specifically, let $\epsilon \in [0,1]$ be the maximum tolerable probability of constraint violation (called \textit{safety level} in~\cite{abate2008probabilistic}),
and choose $\alpha := 1$, $r := \epsilon - \frac{1}{2}$, and $g(x) := \textbf{1}_{\bar{\mathcal{K}}}(x) - \frac{1}{2}$, where
\begin{subequations}\begin{equation}
\textbf{1}_\mathcal{\bar{K}}(x) := \begin{cases} 1 \text{ if } x \notin \mathcal{K} \\ 0 \text{ if } x \in \mathcal{K} \end{cases}.
\end{equation}
Then, the risk-sensitive safe set~\eqref{myS} is equal to, 
\begin{equation}
\Big\{ x \in \mathcal{X} \text{ }\Big|\text{ }  {\underset{\pi \in \Pi}\min}\text{ }\mathbb{E}\big[ \max_{k = 0,\dots, N} \textbf{1}_\mathcal{\bar{K}}(x_k)  \big] \leq \epsilon \Big\},
\label{simpleS}\end{equation}\end{subequations}
which is the \textit{maximal probabilistic safe set} at the $\epsilon$-safety level~\cite[Eqs. 11 and 13]{abate2008probabilistic}, if
we consider non-hybrid dynamic systems that evolve under history-dependent policies.\footnotemark
\footnotetext{Abate et al.~\cite{abate2008probabilistic} considers hybrid dynamic systems that evolve under Markov policies.}

Risk-sensitive safe sets have two desirable mathematical properties.
The first property is that $\mathcal{S}_\alpha^r$ shrinks as the risk level $r$ or the confidence level $\alpha$ decreases.
Since $\mathcal{S}_\alpha^r$ is an $r$-sublevel set and $\text{CVaR}_\alpha$ increases as $\alpha$ decreases,
one can show that $ \mathcal{S}_{\alpha_2}^{r_2} \subseteq \mathcal{S}_{\alpha_1}^{r_2} \subseteq \mathcal{S}_{\alpha_1}^{r_1}$ and $\mathcal{S}_{\alpha_2}^{r_2} \subseteq \mathcal{S}_{\alpha_2}^{r_1} \subseteq \mathcal{S}_{\alpha_1}^{r_1}$ hold for any $r_1 \geq r_2$ and $1 \geq \alpha_1 \geq \alpha_2 > 0$. In other words, as the allowable level of risk of constraint violation ($r$) decreases, or as the fraction of damaging outcomes that are not fully addressed ($\alpha$) decreases,
$\mathcal{S}_\alpha^r$ encodes a higher degree of safety.

The second property is that risk-sensitive safe sets at the risk level, $r := 0$, have probabilistic safety guarantees.
\begin{lemma}[Probabilistic safety guarantee]\label{mylemma1}
If $x \in \mathcal{S}_\alpha^0$, then the probability that the state trajectory initialized at $x$ exits $\mathcal{K}$ can be reduced to $\alpha$ by a control policy.
\end{lemma}
\begin{proof}
The proof follows from the fact that $\text{CVaR}_\alpha[Z_x^\pi] \leq 0 \implies \mathbb{P}[Z_x^\pi\geq 0] \leq \alpha$ 
~\cite[Sec. 6.2.4, pp. 257-258]{shapiro2009lectures}.
%
%
The event $Z_x^\pi\geq 0$ is equivalent to the event that there is a state $x_k$ of the associated trajectory
that exits the constraint set, since $g(x_k) \geq 0 \iff x_k \notin \mathcal{K}$.
%
%
\end{proof}
Lemma~\ref{mylemma1} indicates that $\mathcal{S}_\alpha^0$ is a subset of the \textit{maximal probabilistic safe set} 
at the safety level $\alpha \in (0, 1]$, 
if we consider non-hybrid dynamic systems that evolve under history-dependent policies~\cite[Eqs. 9 and 11]{abate2008probabilistic}.

A key difference between our risk-sensitive safe set~\eqref{myS} and the risk-constrained safe set in~\cite{samuelson2018safety}
is that we specify the CVaR of the worst constraint violation of the state trajectory $(x_0,\dots,x_N)$ to be below a required threshold, while
ref.~\cite{samuelson2018safety} specifies the CVaR of the constraint violation of $x_k$ to be below a required threshold for each $k$.
%
%
\section{Reduction of Risk-Sensitive Safe Set Computation to CVaR-MDP}\label{sec::reduction}
Computing risk-sensitive safe sets is challenging since the computation involves a maximum of costs (as opposed to a summation of costs)
and the Conditional Value-at-Risk measure (as opposed to an expectation). 
In this section, we show how computing an under-approximation of a risk-sensitive safe set can be reduced to solving a CVaR-MDP,
which has been studied, for example, by~\cite{chow2015risk} and~\cite{haskell2015convex}. 
Such a reduction will be leveraged in Section~\ref{sec::alg} to devise a value-iteration algorithm to compute tractable approximations of risk-sensitive safe sets.

\subsection{Preliminaries}
The reduction procedure is inspired by Chow et al.~\cite{chow2015risk}. Specifically, we consider an augmented state space, $\mathcal{X} \times \mathcal{Y}$, that consists of the original state space, $\mathcal{X}$, 
and the space of confidence levels, $\mathcal{Y} := (0,1]$.
The under-approximations of risk-sensitive safe sets will be defined in terms of the dynamics of the augmented state,
$(x,y) \in \mathcal{X} \times \mathcal{Y}$.

Let $(x_0, y_0) = (x, \alpha)$ be a given initial condition.
The augmented state at time $k+1$, $(x_{k+1}, y_{k+1})$, depends on the augmented state at time $k$, $(x_k, y_k)$, as follows.
Given a control $u_k \in U$ and a sampled disturbance $w_k \in D$, the next state $x_{k+1} \in \mathcal{X}$ satisfies the dynamics model~\eqref{sys}. 
The next confidence level $y_{k+1} \in \mathcal{Y}$ is given by,
\begin{equation}
y_{k+1} = \bar{R}_{x_k,y_k}(w_k) \cdot y_k,\label{Rbar}
\end{equation}
where $\bar{R}_{x_k, y_k} : D \to (0, \frac{1}{y_k}]$ is a known deterministic function, 
which we will specify in Lemma~\ref{decomlemma}. 
The augmented state space $\mathcal{X} \times \mathcal{Y}$ is fully observable.
Indeed, the history of states and actions, $(x_0, u_0, \hdots, x_{k-1}, u_{k-1}, x_{k})$, is available at time $k$ by~\eqref{sys}. 
Also, the history of confidence levels, $(y_0, \hdots, y_k)$, is available at time $k$, since
the functions $\bar{R}_{x_k,y_k}$ and the initial confidence level, $y_0 = \alpha$, are known.

We define the set of \textit{deterministic, Markov} control policies in terms of the augmented state space as follows,
\begin{equation}
\label{augpi}
\begin{aligned}
& \bar{\Pi}_t := \{ (\bar{\mu}_t, \bar{\mu}_{t+1}, \dots, \bar{\mu}_{N-1}) \mid \bar{\mu}_k: \mathcal{X} \times \mathcal{Y} \rightarrow U \},\\
& t = 0, \dots, N-1.
\end{aligned}
\end{equation}
There is an important distinction between the set of policies $\bar{\Pi}_0$ as defined above,
and the set of policies $\Pi$ as defined in~\eqref{pi}.
Given $\bar{\pi}_0 \in \bar{\Pi}_0$, the control law at time $k$, $\bar{\mu}_k \in \bar{\pi}_0$, 
only depends on the current state $x_k \in \mathcal{X}$ and the current confidence level $y_k \in \mathcal{Y}$.
However, given $\pi \in \Pi$, the control law at time $k$, $\mu_k \in \pi$, 
depends on the state history up to time $k$, $(x_0, \dots, x_k) \in H_k$.
In particular, the set of policies $\bar{\Pi}_0$ is included in the set of policies $\Pi$.
This is because the augmented state at time $k$ is uniquely determined by the initial confidence level and the state history up to time $k$. 

The benefits of considering $\bar{\Pi}_0$ instead of $\Pi$ are two-fold. 
First, the computational requirements are reduced when the augmented state at time $k$, $(x_k, y_k)$, 
is processed instead of the initial confidence level and the state history up to time $k$, $(y_0, x_0, x_1, \hdots, x_k)$. 
Second, we are able to define an under-approximation
of the risk-sensitive safe set given by~\eqref{myS} using $\bar{\Pi}_0$, which we explain below.
\subsection{Under-Approximation of Risk-Sensitive Safe Set}
Define the set $\mathcal{U}_\alpha^r \subseteq \mathcal{X}$,
at the confidence level $\alpha \in (0,1]$ and the risk level $r \in \mathbb{R}$,
\begin{equation}\label{under}
\mathcal{U}_\alpha^r := \{x \in \mathcal{X} \mid J_0^*(x,\alpha) \leq \beta e^{m\cdot r} \},
\end{equation}
where
\begin{equation}\begin{aligned}
& J_0^*(x,\alpha) := {\underset{\pi \in \bar{\Pi}_0}\min} \text{ CVaR}_\alpha \big[ Y_x^\pi \big],\\
& Y_x^\pi := \textstyle\sum_{k=0}^N c(x_k),
\end{aligned}\label{J0}\end{equation}
such that $c : \mathcal{X} \to \mathbb{R}$ is a stage cost, and the augmented state trajectory, $(x_0, y_0, \dots, x_{N-1}, y_{N-1}, x_N)$,
satisfies~\eqref{sys} and~\eqref{Rbar} with the initial condition $(x_0, y_0) = (x, \alpha)$ under the policy $\pi \in \bar{\Pi}_0$.
The next theorem, whose proof is provided in the Appendix, states that if the stage cost takes a particular form, then $\mathcal{U}_\alpha^r$ is an under-approximation of the risk-sensitive safe set $\mathcal{S}_\alpha^r$.   
\begin{theorem}[Reduction to CVaR-MDP]\label{lemma2}
Choose the stage cost $c(x) := \beta e^{m\cdot g(x)}$, where $\beta > 0$ and $m > 0$ are constants, and $g$ satisfies~\eqref{g}.
Then, $\mathcal{U}_\alpha^r$ as defined in~\eqref{under} is a subset of $\mathcal{S}_\alpha^r$ as defined in~\eqref{myS}. 
Further, the gap between $\mathcal{U}_\alpha^r$ and $\mathcal{S}_\alpha^r$ can be reduced by increasing $m$.\hspace{6em}$\blacksquare$ \end{theorem}

In the definition of the stage costs, the parameter $\beta$ is included to address numerical issues that may arise, if $m$ is set to a very large number.
\section{A Value-Iteration Algorithm \\to Approximate Risk-Sensitive Safe Sets}\label{sec::alg}
By leveraging Theorem \ref{lemma2}, one can use existing CVaR-MDP algorithms to compute under-approximations of risk-sensitive safe sets. 
In this paper, we adapt a value-iteration algorithm from Chow et al.~\cite{chow2015risk} to compute tractable approximations 
of the risk-sensitive safe set under-approximations $\{\mathcal{U}_{\alpha}^r\}$. 
We start by stating an existing temporal decomposition result for CVaR that will be instrumental to devising the value-iteration algorithm. 
\subsection{Temporal Decomposition of Conditional Value-at-Risk}
In this section, we present an existing result (namely, Lemma 22 in \cite{pflug2016time}) 
that specifies how the Conditional Value-at-Risk of a sum of costs can be partitioned over time, and
how the confidence level evolves over time, which motivates the choice of the update function~\eqref{Rbar}.
\begin{lemma}[Temporal decomposition of CVaR]\label{decomlemma}
At time $k$, suppose that the system~\eqref{sys} is at the state $x_k \in \mathcal{X}$ with the confidence level $y_k \in \mathcal{Y}$ 
and is subject to a policy $\pi_k := (\mu_k, \pi_{k+1}) \in \bar{\Pi}_k$. Then,
\begin{subequations}\label{decomp}
\begin{equation}\begin{aligned}
& \text{CVaR}_{y_k} [ Z | x_k, \pi_k ] = {\underset{R \in \mathcal{R}(y_k, \mathbb{P})}\max} C(R, Z; x_k, y_k, \pi_k), \\
& C(R, Z; x_k, y_k, \pi_k) :=\\
& \mathbb{E}_{w_k \sim \mathbb{P}}\big[ R(w_k) \cdot \text{CVaR}_{y_k R(w_k)}[ Z | x_{k+1}, \pi_{k+1} ] \big| x_k, \mu_k \big],
\end{aligned}
\end{equation}
where 
\begin{equation}\begin{aligned}
& \mathcal{R}(y_k, \mathbb{P})
:= \big\{ R : D \to \big(0,\textstyle\frac{1}{y_k}\big] \mathrel{\big|} \mathbb{E}_{w_k \sim \mathbb{P}}\big[ R(w_k) \big] = 1 \big\}, \\
& Z := \textstyle \sum_{i=k+1}^N c(x_i),
\end{aligned}\end{equation}
\end{subequations}
such that $c: \mathcal{X} \to \mathbb{R}$ is a stage cost.
Further, given the current state $(x_k, y_k)$, the current control $u_k := \mu_k(x_k, y_k)$, and the next state $x_{k+1}$, 
the function that was introduced in~\eqref{Rbar} $\bar{R}_{x_k,y_k} : D \to (0,\frac{1}{y_k}]$ is defined as,
\begin{equation}\begin{aligned}
& \bar{R}_{x_k,y_k}(w_k) = {\underset{R \in \mathcal{R}(y_k,\mathbb{P})}\argmax} C(R, Z; x_k, y_k, \pi_k).
\end{aligned}
\end{equation}
\end{lemma}
\begin{remark}
The proof of Lemma~\ref{decomlemma} is a consequence of Lemma 22 in~\cite{pflug2016time}, and its proof is omitted for brevity. 
\end{remark}
\begin{remark}
If we do not have access to $w_k$, but only to $(x_k, y_k, u_k, x_{k+1})$, then the next confidence level is defined as 
$y_{k+1} := \bar{R}_{x_k, y_k}(w)$, where $w \in D$ is any disturbance that satisfies $x_{k+1} = f(x_k, u_k, w)$.
\end{remark}
\begin{remark}
$\text{CVaR}_{y_k} [ Z | x_k, \pi_k ]$ is the risk of 
the cumulative cost of the trajectory, $(x_{k+1}, \dots, x_N)$, that is initialized at the state $x_k$ 
with the confidence level $y_k$ and is subject to the policy $\pi_k \in \bar{\Pi}_k$.
\end{remark}

\subsection{Value-Iteration Algorithm}
Using Lemma~\ref{decomlemma}, we will devise a dynamic programming value-iteration algorithm to compute an approximation $J_0$ of $J_0^*$, 
and thus, an approximation of $\mathcal{U}_\alpha^r$ 
at different levels of confidence $\alpha$ and risk $r$.

Specifically, compute the functions $J_{N-1}$, \dots, $J_0$ recursively as follows: for all $z_k := (x_k, y_k) \in \mathcal{X} \times \mathcal{Y}$,
\begin{equation}\begin{aligned}
& J_k(z_k) \\
&:= {\underset{u \in U}\min} \Big\{ c(x_k) + {\underset{R \in \mathcal{R}(y_k, \mathbb{P})}\max} \mathbb{E}_{w_k \sim \mathbb{P}}\big[ R J_{k+1}(x', y_k R) \big| z_k, u \big] \Big\}, \\
& \text{ for }k = N-1, \dots, 0,
\label{bell}\end{aligned}\end{equation}
where $J_N(x_N, y_N) := c(x_N)$, $c(x) := \beta e^{m \cdot g(x)}$, $x' := x_{k+1}$ satisfies~\eqref{sys}, 
and $\mathcal{R}(y_k, \mathbb{P})$ is defined in~\eqref{decomp}. 

Then, we approximate the set $\mathcal{U}_{\alpha}^r$ as $\widehat{\mathcal{U}}_{\alpha}^r := \left\{ x \in \mathcal{X} \mid J_0(x, \alpha) \leq \beta e^{m \cdot r}\right\}$, 
where we have replaced $J_0^*$ in~\eqref{under} with $J_0$. The function, $J_0$, is obtained from the last step of the value iteration~\eqref{bell}. 

In the Appendix, we present theoretical arguments inspired by~\cite{chow2015risk} and~\cite[Sec. 1.5]{bertsekas2005dynamic} 
that justify such an approximation. In particular, we provide theoretical evidence for the following conjecture.  

{\bf Conjecture (C):}
Assume that the functions $J_{N-1}$, \dots, $J_0$ are computed recursively as per the value-iteration algorithm \eqref{bell}. 
Then, for any $(x, \alpha) \in \mathcal{X} \times \mathcal{Y}$,
\begin{equation}
\label{thm:eq}
J_0(x,\alpha) = J_0^*(x, \alpha),
\end{equation}
where $J_0^*$ is given by~\eqref{J0}.

This conjecture is further supported by numerical experiments presented next.
%
\section{Numerical Experiments}\label{sec::ex}
In this section, we provide empirical results that demonstrate the following: 1) our value-iteration estimate of $J_0$ is close to a Monte Carlo estimate of $J_0^*$,
2) our value-iteration estimate of $\widehat{\mathcal{U}}_y^r$ is an under-approximation of a Monte Carlo estimate of $\mathcal{S}_y^r$, and
3) estimating $J_0$ (and $\widehat{\mathcal{U}}_y^r$) via the value-iteration algorithm is tractable on a
realistic example inspired from stormwater catchment design.
Item 1 provides empirical support for the Conjecture. Items 2 and 3 provide empirical support for 
reducing the computation of risk-sensitive safe sets to a CVaR-MDP.
In our experiments, we used MATLAB R2016b (The MathWorks, Inc., Natick, MA) and MOSEK (Copenhagen, Denmark) with CVX~\cite{grant2008cvx} on a
standard laptop (64-bit OS, 16GB RAM, Intel\textsuperscript{\textregistered} Core\textsuperscript{TM} i7-4700MQ CPU @ 2.40GHz).
Our code is available at~\url{https://github.com/chapmanmp/ACC_2019_Github}.

This section demonstrates the utility of computing approximate risk-sensitive safe sets in a practical setting:
to evaluate the design of a retention pond in a stormwater catchment system. We consider a retention pond from our prior work~\cite{sustech} as a stochastic discrete-time dynamic system,
\begin{equation}\begin{aligned}
& x_{k+1} = x_k + \frac{\triangle t}{A} (w_k - q_p(x_k, u_k)), \text{ }k = 0, \dots, N-1, \\
& q_p(x_k,u_k) := \begin{cases} C_d \pi r^2 u_k \sqrt{ 2\eta(x-E) } & \text{ if } x_k \geq E \\
						0 & \text{ if } x_k < E, \end{cases}
\end{aligned}\label{watersys}\end{equation}
where $x_k \geq 0$ is the water level of the pond in feet at time $k$, $u_k \in U := \{0, 1\}$ is the valve setting at time $k$,
and $w_k \in D := \{d_1, \dots, d_{10}\}$ is the random surface runoff in feet-cubed-per-second at time $k$. ($\eta = 32.2\frac{\text{ft}}{\text{s}^2}$ is the acceleration due to gravity, $\pi \approx 3.14$, $r = \frac{1}{3}$ft is the outlet radius, $A = 28,292$ft\textsuperscript{2} is the pond surface area, $C_d = 0.61$ is the discharge coefficient,
and $E = 1$ft is the elevation of the outlet.)
We estimated a finite probability distribution for $w_k$ using the surface runoff samples that we previously generated from a time-varying \textit{design storm}
(a synthetic storm based on historical rainfall)~\cite{sustech}. 
We averaged each sample over time and solved for a distribution that satisfied the empirical statistics of the time-averaged samples (Table~\ref{dist}). 
We set $\triangle t := 300$ seconds, and $N := 48$ to yield a 4-hour time horizon.
We chose the constraint set $\mathcal{K} := [0, 5)$, and $g(x) := x - 5$.
\begin{table}[b]
\vspace{-2em}
\begin{center}
\caption{}
\begin{tabular}{| p{3.5cm} | p{3.5cm} |}
\hline
\bf{Sample moment} & \bf{Value}  \\ \hline
Mean & 12.16 ft\textsuperscript{3}/s \\ 
Variance & 3.22 ft\textsuperscript{6}/s\textsuperscript{2} \\ 
Skewness & 1.68 ft\textsuperscript{9}/s\textsuperscript{3} \\ 
\hline 
\textbf{Disturbance sample}, $d_j$ ft\textsuperscript{3}/s & \textbf{Probability}, $\mathbb{P}[w_k = d_j]$ \\ \hline
$8.57$ 		& $0.0236$ \\
$9.47$ 		& $10^{-4}$ \\
$10.37$ 		& $10^{-4}$ \\
$11.26$  & $0.5249$ \\ 
$12.16$ & $0.3272$ \\ 
$13.06$  & $10^{-4}$ \\ 
$13.95$  & $10^{-4}$ \\ 
$14.85$  & $10^{-4}$ \\ 
$15.75$  & $10^{-4}$ \\ 
$16.65$  & $0.1237$ \\ \hline
\end{tabular}
\begin{flushleft} \end{flushleft}
\label{dist}
\end{center}
\vspace{-2em}
\end{table}
We computed over a grid of states and confidence levels $G := G_s \times G_c$, where  
$G_s := \{0, 0.1, \dots, 6.4, 6.5\}$,
and $G_c := \{0.999, 0.95, 0.80, 0.65, 0.5, 0.35, 0.20, 0.05, 0.001\}$.
Since the initial state $x_0$ is non-negative and the smallest realization of $w_k$ is about 8.5$\frac{\text{ft}^3}{\text{s}}$, 
$x_{k+1}\geq x_k$ for all $k$. 
If $x_{k+1} > 6.5$ft, we set $x_{k+1} := 6.5$ft to stay within the grid.

We were able to empirically assess the accuracy of our proposed approach because an optimal control policy is known \textit{a priori} for the one-pond system.
Since $x_{k+1}\geq x_k$ for all $k$, and the only way to exit the constraint set is if $x_k \geq 5$ft,
an optimal policy is to keep the valve open over all time, regardless of the current state, the current confidence level, or the state history up to the current time.

Our value-iteration estimate of the function $J_0$ is shown in Fig.~\ref{J0dp}, and a Monte Carlo estimate of the function $J_0^*$ is shown in Fig.~\ref{J0mc}.
The estimates of $J_0$ and $J_0^*$ are similar throughout the grid except near the smaller confidence levels.
The average (largest) difference normalized by the estimate of $J_0^*$  is approximately $1.4$ ($18.7$).
The average (largest) difference normalized by the estimate of $J_0$  is approximately $0.23$ ($0.95$).
These results provide empirical support for the Conjecture.
 
Our value-iteration estimate of the set $\widehat{\mathcal{U}}_{\alpha}^r$ and a Monte Carlo estimate of the set $\mathcal{S}_y^r$ are shown in Fig.~\ref{compare} at different levels of confidence $y$ and risk $r$. 
The empirical results indicate that $\widehat{\mathcal{U}}_{\alpha}^r$ is an under-approximation of $\mathcal{S}_y^r$. We estimated the sets $\{\mathcal{S}_y^r\}$ using a Monte Carlo estimate of the function $W_0^*$, which is shown in Fig.~\ref{W0mc}.

The computation time for our value-iteration estimate of $J_0$ was roughly 3h 6min.
We deem this performance to be acceptable because  
1) computations to evaluate design choices are performed off-line, 
2) the problem entailed a realistically sized state space ($|G_s|\cdot|G_c| = 594$ grid points) and time horizon ($N = 48$ time points),
and 3) our implementation is not yet optimized. Further, there is recent work in scalable approximations of reachable sets (e.g.,~\cite{EECS-2018-41}) that we will investigate for possible extensions to the risk-sensitive case.
\begin{figure}[thpb]
      \centering
      \includegraphics[scale=0.5]{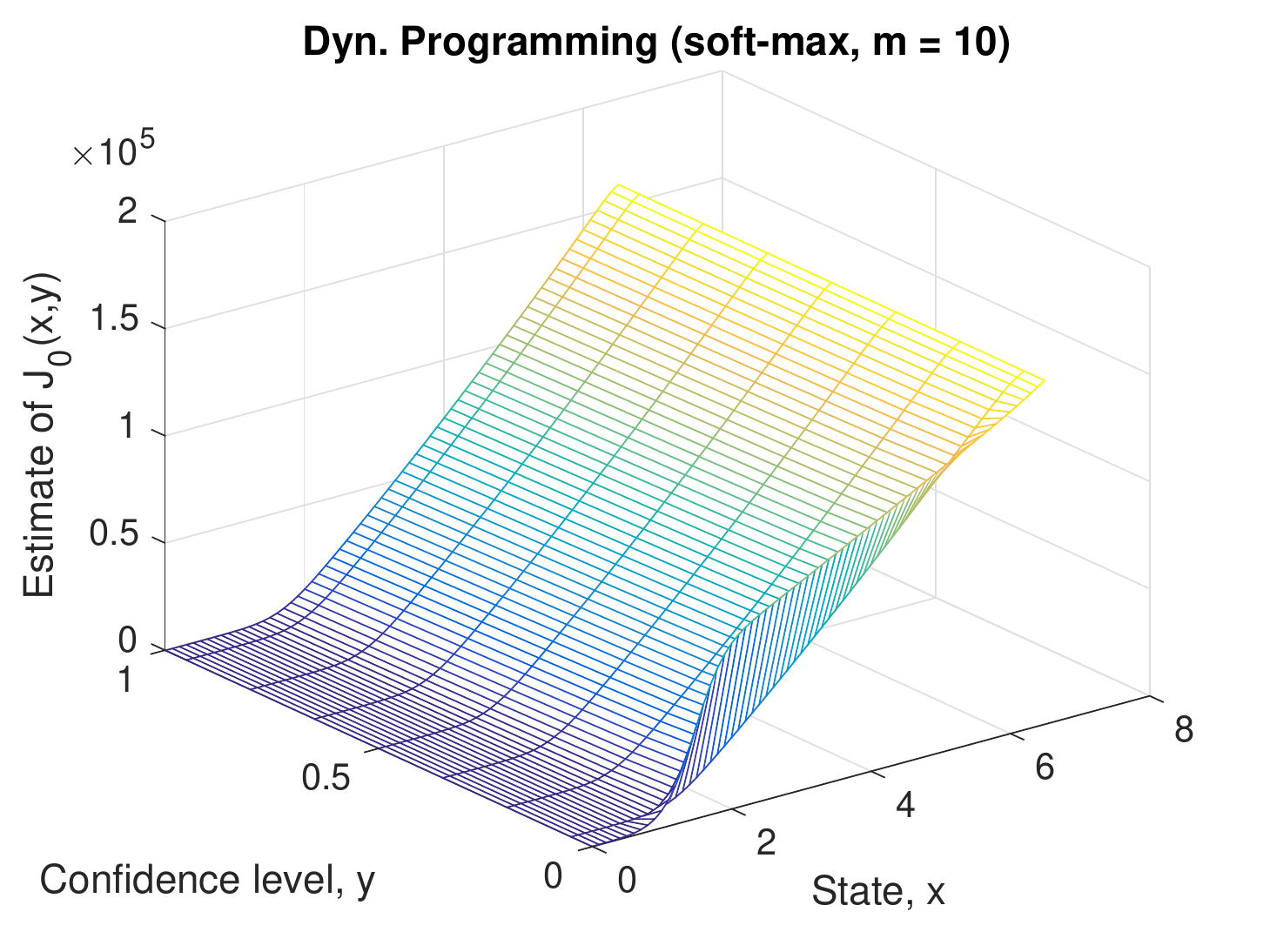}
      \vspace{-1em}
      \caption{Our value-iteration estimate of $J_0(x,\alpha)$ versus $(x, \alpha) \in G$ for the pond system, see~\eqref{bell}.
	  $c(x) := \beta e^{m\cdot g(x)}$, $\beta := 10^{-3}$, $m := 10$, and $g(x) := x-5$. The computation time was roughly 3h 6min.}
      \label{J0dp}
\end{figure}
\begin{figure}[thpb]
      \centering
      \includegraphics[scale=0.5]{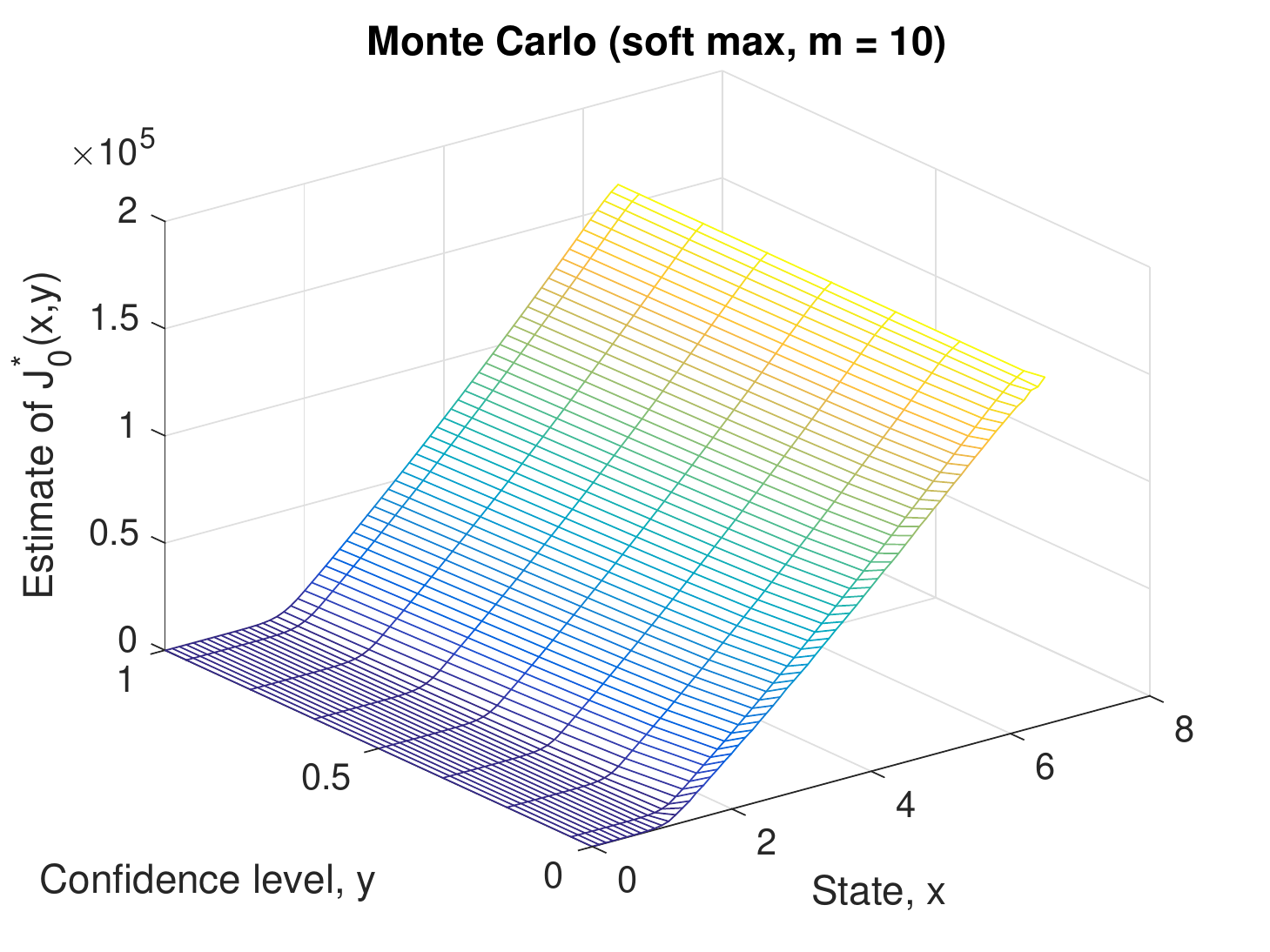}
      \vspace{-1em}
      \caption{A Monte Carlo estimate of $J_0^*(x, \alpha)$ versus $(x, \alpha) \in G$ for the pond system.
	  $c(x) := \beta e^{m \cdot g(x)}$, $\beta := 10^{-3}$, $m := 10$, and $g(x) := x-5$.
	  100,000 samples were generated per grid point. See also Fig.~\ref{J0dp}.}
      \label{J0mc}
\end{figure}
\begin{figure}[thpb]
\vspace{-.6cm}
      \centering
      \includegraphics[scale=0.45]{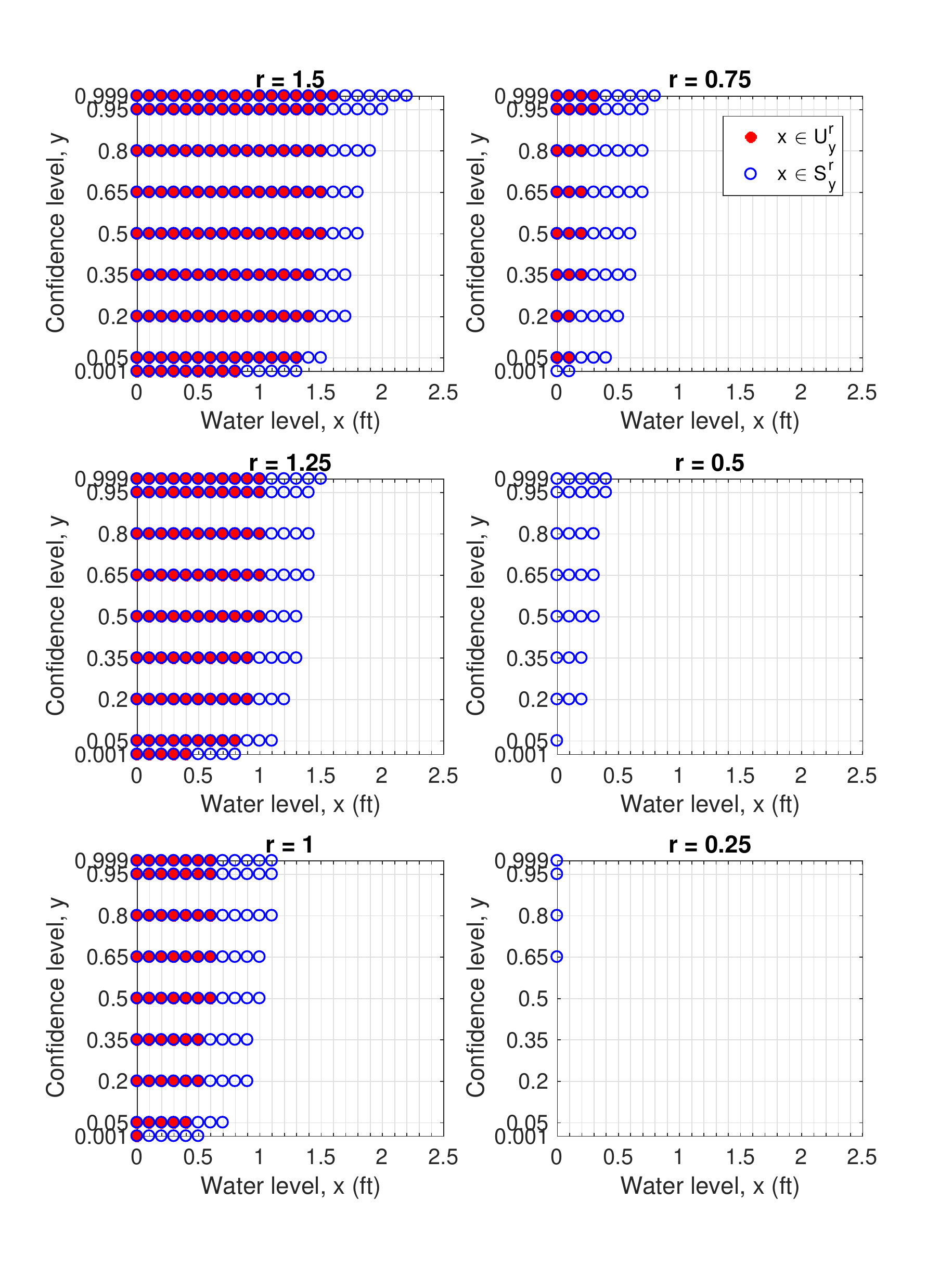}
      \vspace{-2em}
      \caption{ Approximations of $\{\widehat{\mathcal{U}}_y^r\}$ and $\{\mathcal{S}_y^r\}$ are shown for the pond system at various levels of confidence $y$ and risk $r$. In the legend,  $\widehat{\mathcal{U}}_y^r$ is denoted by $\text{U}_y^r$, and $\mathcal{S}_y^r$ is denoted by $\text{S}_y^r$. Approximations of $\{\widehat{\mathcal{U}}_y^r\}$ were obtained from our value-iteration estimate of $J_0$ (see Fig.~\ref{J0dp}). Approximations of $\{\mathcal{S}_y^r\}$ were obtained from a Monte Carlo estimate of $W_0^*$ (see Fig.~\ref{W0mc}).}\label{compare}
\end{figure}
\begin{figure}[thpb]
      \centering
      \includegraphics[scale=0.5]{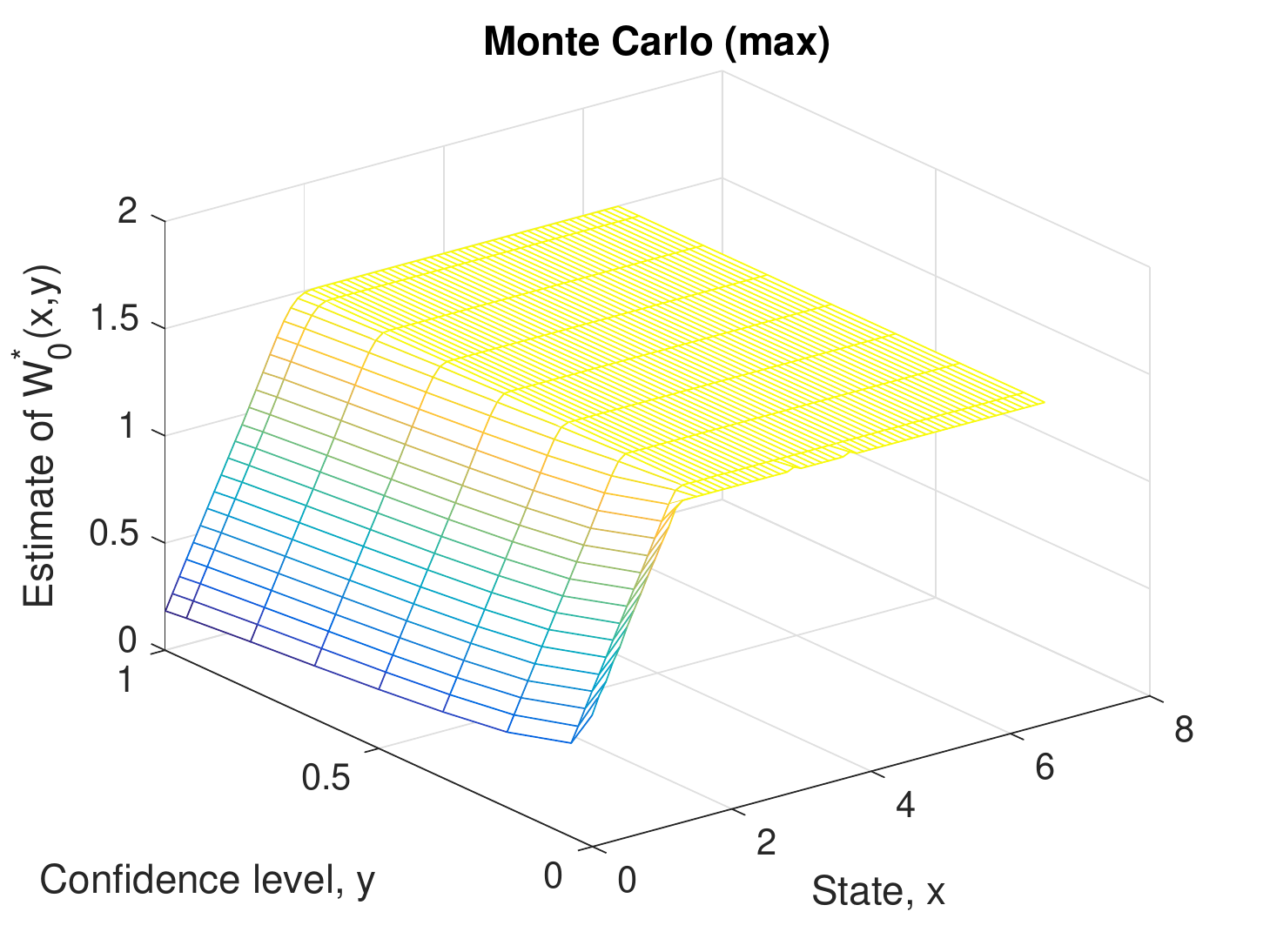}
      \vspace{-1em}
      \caption{A Monte Carlo estimate of $W_0^*(x,\alpha)$, as defined in~\eqref{myS}, versus $(x, \alpha) \in G$ for the pond system.
	  100,000 samples were generated per grid point, and $g(x) := x - 5$. 
	  The maximum is approximately 1.5ft because the system state was prevented from exceeding 6.5ft.}
      \label{W0mc}
\end{figure}

\textit{Value-iteration implementation.} To implement the value-iteration algorithm, we used the interpolation method over the confidence levels proposed by Chow et al.~\cite{chow2015risk} 
to approximate the expectation in~\eqref{bell} as a piecewise linear concave function, which we maximized by solving a linear program.
Further, at each $\alpha \in G_c$, we used multi-linear interpolation to approximate the value of $J_{k+1}(x_{k+1}, \alpha)$. We set $J_{k+1}(x_{k+1}, \alpha) := \frac{(x_{k+1} - x_i) \cdot J_{k+1}(x^{i+1}, \alpha) + (x^{i+1} - x_{k+1})  \cdot J_{k+1}(x^i, \alpha)}{x_{i+1}-x_i}$,
where $x^i \in G_s$ and $x^{i+1} \in G_s$ are the two nearest grid points to $x_{k+1}$ that satisfy $x^i \leq x_{k+1} \leq x^{i+1}$.

\textit{Monte Carlo implementations.}
For each $(x,\alpha) \in G$, we sampled 100,000 trajectories starting from $x_0 = x$, subject to keeping the valve open over time,
as this is an optimal policy.
For each trajectory sample $i$, we computed the cost sample $z_i := \max_k\{g(x_k^i)\}$, and estimated the Conditional Value-at-Risk
of the 100,000 cost samples at the confidence level $\alpha$. 
We used the CVaR estimator $\widehat{\text{CVaR}}_\alpha[Z] := \frac{1}{\alpha M}\sum_{i=1}^M z_i \textbf{1}_{\{z_i\geq \hat{Q}_\alpha\}}$,
where $\hat{Q}_\alpha$ is the $(1-\alpha)$-quantile of the empirical distribution of the samples $\{z_i\}_{i=1}^M$,
and $M := 100,000$ is the number of samples~\cite[Sec. 6.5.1]{shapiro2009lectures}.
Since this estimator is designed for continuous distributions, 
we added zero-mean Gaussian noise with a small standard deviation, $\sigma := 10^{-12}$, to each cost sample prior to computing the CVaR.
Fig.~\ref{W0mc} provides a Monte Carlo estimate of $W_0^*$.
To obtain a Monte Carlo estimate of $J_0^*$, we used the same procedure but with the cost sample 
$z_i := \xi_i + \sum_{k=0}^N \beta e^{m\cdot g(x_k^i)}$, $\xi_i \sim \mathcal{N}(0, \sigma := 10^{-7})$,
$m := 10$, and $\beta := 10^{-3}$; see Fig.~\ref{J0mc}.
%
%
\section{Conclusion}\label{conc}
In this paper, we propose the novel idea of a risk-sensitive safe set
to encode safety of a stochastic dynamic system in terms of an allowable level of risk of constraint violations $r$ 
in the $\alpha$-fraction of the most damaging outcomes.
We show how the computation of a risk-sensitive safe set can be
reduced to the solution to a Markov Decision Process, where cost is assessed according to the Conditional Value-at-Risk measure. 
Further, we devise a tractable algorithm to approximate a risk-sensitive safe set, and provide theoretical and empirical arguments about its correctness.

Risk-sensitive safe sets have the potential to inform the design of safety-critical infrastructure systems,
by revealing trade-offs between the risk of damaging outcomes and design choices at different levels of confidence.
We illustrate our risk-sensitive reachability approach on a stormwater retention pond that must be designed to operate safely in the presence of uncertain rainfall. 
Our results reveal that the current design of the pond is likely undersized: even if the pond starts empty, 
there is a risk of at least 0.25ft of overflow at most levels of confidence 
under the random surface runoff of the design storm (see Fig.~\ref{compare}, $r=0.25$ plot at $x=0$).

On the methodological side, future steps are: 1) to formally prove the correctness of the value-iteration algorithm, 
2) to devise approximate value-iteration algorithms to improve scalability, and 3) to consider a broader class of risk measures. 
On the applications side, future steps are: 1) to adjust the parameters of the dynamics model (e.g., outlet radius) to reduce the risk of extreme overflows, 
2) to apply our method to a more realistic stormwater system that consists of two ponds in series on a larger grid,
and 3) to develop an optimized toolbox for the computation of risk-sensitive safe sets.  
We are hopeful that with further development, the concept of risk-sensitive reachability will become a valuable tool for the design of safety-critical systems.

\section*{ACKNOWLEDGMENTS}
We thank Sumeet Singh, Dr. Mo Chen, Dr. Murat Arcak, Dr. Alessandro Abate, and Dr. David Freyberg for discussions. 
MC is supported by an NSF Graduate Research Fellowship and was supported by a Berkeley Fellowship for Graduate Studies. 
This work is supported by NSF CPS 1740079 and NSF PIRE UNIV59732.

\section*{APPENDIX}\label{appendix}

\begin{proof}[Proof of Theorem~\ref{lemma2}]
The proof relies on two facts. The first fact is,
\begin{subequations}\label{related}
\begin{equation}\begin{aligned}
\max\{y_1, \dots, y_p\} & \leq \frac{1}{m}\log( e^{my_1}+ \dots +e^{my_p} ) \\
						& \leq \max\{y_1, \dots, y_p\} + \frac{\log p}{m},
\end{aligned}\end{equation}
for any $y \in \mathbb{R}^p$, $m > 0$. (Use the log-sum-exp relation stated in~\cite[Sec. 3.1.5]{boyd2004convex}.) So, as $m \rightarrow \infty$,
\begin{equation}
\frac{1}{m}\log( e^{my_1}+ \dots +e^{my_p} ) \rightarrow \max\{y_1, \dots, y_p\}.
\end{equation}
\end{subequations}
The second fact is that CVaR is a \textit{coherent risk measure},
so it satisfies certain properties. 
CVaR is \textit{positively homogeneous}, $\text{CVaR}_\alpha[\lambda Z] = \lambda\text{CVaR}_\alpha[Z]$ 
for any $\lambda \geq 0$,
and \textit{monotonic}, $\text{CVaR}_\alpha[Y] \leq \text{CVaR}_\alpha[Z]$ for any random variables $Y \leq Z$~\cite[Sec. 2.2]{kisiala2015conditional}.
Also, CVaR can be expressed as the supremum expectation over a particular set of probability density functions~\cite[Eqs. 6.40 and 6.70]{shapiro2009lectures}.
Using this property and $\mathbb{E}[\log(Z)] \leq \log \left(\mathbb{E}[Z]\right)$, one can show,
\begin{equation} \text{CVaR}_\alpha[\log(Z)] \leq \log \left(\text{CVaR}_\alpha[Z]\right), \label{logeq}\end{equation}
for any random variable $Z$ with finite expectation. 

By monotonicity, positive homogeneity,~\eqref{related}, and~\eqref{logeq},
\begin{equation}\begin{aligned}
\text{CVaR}_\alpha\big[ Z_x^\pi \big] & \leq \textstyle\frac{1}{m} \text{CVaR}_\alpha\big[ \log\left( \bar{Y}_x^\pi \right) \big] \\
& \leq \textstyle\frac{1}{m} \log \left(\text{CVaR}_\alpha\big[ \bar{Y}_x^\pi \big] \right),
\end{aligned}\label{12}\end{equation}
where $\bar{Y}_x^\pi := Y_x^\pi/\beta$. Now, if $x \in \mathcal{U}_\alpha^r$, then
\begin{equation*}
e^{m\cdot r} \geq {\underset{\pi \in \bar{\Pi}_0}\min}\text{ CVaR}_\alpha \big[ Y_x^\pi /\beta \big] \geq {\underset{\pi \in \Pi}\min}\text{ CVaR}_\alpha \big[ Y_x^\pi /\beta \big], 
\end{equation*}
since $\bar{\Pi}_0$ is included in $\Pi$. By Lemma~\ref{lemma::infeqmin}, there exists $\pi \in \Pi$ such that
\begin{equation*}
r \geq \textstyle\frac{1}{m}\log \left( \text{CVaR}_\alpha \big[  Y_x^\pi /\beta \big] \right) \geq \text{CVaR}_\alpha\big[ Z_x^\pi \big], \\
\end{equation*}
where the second inequality holds by~\eqref{12}. So, $x \in \mathcal{S}_\alpha^r$.
\end{proof}
\vspace{1em}
\textit{Theoretical Justification of Conjecture (C):}
Let $\epsilon > 0$. For all $k = 0, \dots, N-1$ and $z_k := (x_k, y_k) \in \mathcal{X} \times \mathcal{Y}$, let $\mu_k^\epsilon : \mathcal{X} \times \mathcal{Y} \to U$ satisfy,
\begin{equation} 
c(x_k) + {\underset{R \in \mathcal{R}(y_k, \mathbb{P})}\max} \mathbb{E}[ RJ_{k+1}(x_{k+1}, y_k R) | z_k, \mu_k^\epsilon] \leq J_k(z_k) + \epsilon.
\label{first}\end{equation} 
Let $J_k^\epsilon$ be a sub-optimal cost-to-go starting at time $k$,
\begin{equation} 
J_k^\epsilon(z_k) := \text{CVaR}_{y_k}\big[\textstyle \sum_{i=k}^N c(x_i) \big| z_k, \pi_k^\epsilon \big],
\label{Jkeps}\end{equation}
where $\pi_k^\epsilon := (\mu_k^\epsilon,\dots,\mu_{N-1}^\epsilon) = (\mu_k^\epsilon, \pi_{k+1}^\epsilon)$. Recall $J_k$, as defined in~\eqref{bell}. Define $J_k^*(z_k) := \min_{\pi \in \bar{\Pi}_k} \text{CVaR}_{y_k}\big[\textstyle \sum_{i=k}^N c(x_i) \big| z_k, \pi \big]$.
To prove the Conjecture, we would like to show  by induction that for all $z_k := (x_k, y_k) \in \mathcal{X} \times \mathcal{Y}$ and $k = N-1, \dots, 0$,
\begin{subequations}\label{abc}
\begin{equation}
J_k(z_k) \leq J_k^\epsilon(z_k) \leq J_k(z_k) + (N-k)\epsilon, 
\label{a}\end{equation}
\begin{equation}
J_k^*(z_k) \leq J_k^\epsilon(z_k) \leq J_k^*(z_k) + (N-k)\epsilon, 
\label{b}\end{equation}
\begin{equation}
J_k(z_k) = J_k^*(z_k), 
\label{c}\end{equation}
\end{subequations}
which is the proof technique in~\cite[Sec. 1.5]{bertsekas2005dynamic}. 
One can show~\eqref{abc} for the base case, $k := N-1$, since $J_N$ is known.
Assuming~\eqref{abc} holds for index $k+1$ (induction hypothesis), we want to show that~\eqref{abc} holds for index $k$ (induction step). The key 
idea is to use the recursion,
\begin{equation}
J_k^\epsilon(z_k) = c(x_k) + {\underset{R \in \mathcal{R}(y_k, \mathbb{P})}\max} \mathbb{E}[ R\cdot J^\epsilon_{k+1}(x_{k+1}, y_k R) | z_k, \mu_k^\epsilon],
\label{Jkepsrec}\end{equation}
which we justify next. Let $Z := \sum_{i=k+1}^N c(x_i)$. 
\begin{equation*}\begin{aligned}
J_k^\epsilon(z_k)& - c(x_k) = \text{CVaR}_{y_k}\big[ Z \big| z_k, \pi_k^\epsilon \big]\\
& = {\underset{R \in \mathcal{R}(y_k, \mathbb{P})}\max}\mathbb{E}\big[ R\cdot \text{CVaR}_{y_k R}[ Z | x_{k+1}, \pi^\epsilon_{k+1} ] \big| z_k, \mu_k^\epsilon \big] \\
& {\overset{(a)}{=}} {\underset{R \in \mathcal{R}(y_k, \mathbb{P})}\max}\mathbb{E}\big[ R\cdot J_{k+1}^\epsilon(x_{k+1},y_k R) \big| z_k, \mu_k^\epsilon \big],
\end{aligned}\end{equation*}
where we use~\eqref{decomp},~\eqref{Jkeps}, and \textit{translation invariance} ($\text{CVaR}_\alpha[a+Z] = a + \text{CVaR}_\alpha[Z]$ for $a \in \mathbb{R}$, see~\cite[Sec. 2.2]{kisiala2015conditional}). 
The last equality (a) is {\em the crux of the Conjecture}, 
as one needs to justify why the worst-case density $R$ is equal to the \textit{a priori} chosen density $\bar{R}$ that defines the dynamics of the confidence level. 
Based on \cite{chow2015risk}, we believe this equality to be correct, but we leave its formal proof for future research. 
Assuming the aforementioned equality is correct, then we show~\eqref{a} for index $k$ using~\eqref{Jkepsrec} and the induction hypothesis. 
Let $\bar{\epsilon}_k := (N-k-1)\epsilon$, and $x' := x_{k+1}$.
\begin{equation*}\begin{aligned}
J_k^\epsilon(z_k) & \leq c(x_k) + {\underset{R \in \mathcal{R}(y_k, \mathbb{P})}\max} \mathbb{E}\big[ R\big(J_{k+1}(x',y_k R) + \bar{\epsilon}_k\big) \big| z_k, \mu_k^\epsilon \big]\\
& = c(x_k) + {\underset{R \in \mathcal{R}(y_k, \mathbb{P})}\max} \mathbb{E}\big[ R J_{k+1}(x_{k+1},y_k R) \big| z_k, \mu_k^\epsilon \big] + \bar{\epsilon}_k\\ 
& \leq J_k(z_k) + (N-k)\epsilon,
\end{aligned}\end{equation*}
since $\mathbb{E}[R] = 1$, and by~\eqref{first}. By~\eqref{bell}, sub-optimality of $\mu_k^\epsilon(z_k) \in U$, $J_{k+1} \leq J^\epsilon_{k+1}$, and~\eqref{Jkepsrec},
\begin{equation*}\begin{aligned}
J_k(z_k) & \leq c(x_k) + {\underset{R \in \mathcal{R}(y_k, \mathbb{P})}\max} \mathbb{E}\big[ R J^\epsilon_{k+1}(x_{k+1},y_k R) \big| z_k, \mu_k^\epsilon \big]\\
& = J_k^\epsilon(z_k), \end{aligned}\end{equation*}
which would complete the induction step for~\eqref{a}, if~\eqref{Jkepsrec} holds.
Next, we show~\eqref{b} for index $k$. By definition, $J_k^*$ is the optimal risk-sensitive cost-to-go from stage $k$, thus, $J_k^* \leq J_k^\epsilon$.
Let $\hat{\epsilon}_k := (N-k)\epsilon$, $x':=x_{k+1}$, $y' := y_kR$, and $Z := \sum_{i=k+1}^N c(x_i)$. For any $\pi_k := (\mu_k, \pi') \in \bar{\Pi}_k$,
\begin{equation*}\begin{aligned}
J_k^\epsilon(z_k) & \leq J_k(z_k) + \hat{\epsilon}_k \\
& \leq c(x_k) + {\underset{R \in \mathcal{R}(y_k, \mathbb{P})}\max} \mathbb{E}\big[ RJ_{k+1}^*(x_{k+1},y_kR) \big| z_k, \mu_k \big] + \hat{\epsilon}_k\\
& \leq c(x_k) + {\underset{R \in \mathcal{R}(y_k, \mathbb{P})}\max} \mathbb{E}\big[ R\text{CVaR}_{y'} [Z|x', \pi' ] \big| z_k, \mu_k \big] + \hat{\epsilon}_k\\
& = c(x_k) + \text{CVaR}_{y_k}[Z|z_k, \pi_k] + \hat{\epsilon}_k\\
& = \text{CVaR}_{y_k}\big[\textstyle \sum_{i=k}^N c(x_k)|z_k, \pi_k\big] + (N-k)\epsilon.
\end{aligned}\end{equation*}
The above statement implies
\begin{equation*}\begin{aligned}
J_k^\epsilon(z_k) & \leq {\underset{\pi \in \bar{\Pi}_k} \min} \text{ CVaR}_{y_k}\big[\textstyle \sum_{i=k}^N c(x_k)|z_k, \pi_k\big] + (N-k)\epsilon \\
& = J_k^*(z_k) + (N-k)\epsilon, \\
\end{aligned}\end{equation*}
which completes the induction step for~\eqref{b}. 

Thus, if~\eqref{Jkepsrec} holds, then~\eqref{a} and~\eqref{b} hold for index $k$ for any $\epsilon>0$. So,~\eqref{c} would hold for index $k$. 
Assuming the conjectured equality (a) is correct, this would complete the proof of the Conjecture.

\section*{AUTHOR CONTRIBUTIONS}
MC: formulation of risk-sensitive safe set based on CVaR so that it generalized stochastic reachability, formulation and proof of reduction to CVaR-MDP theorem (Theorem 1), identified minimax error in Chow 2015 paper, theoretical justification for value iteration using techniques from~\cite{bertsekas2005dynamic} (see extended manuscript), implementation of numerical example, connected experts in controls, risk measures, stormwater catchment design together to create the paper; wrote initial draft and revised to incorporate co-authors’ feedback.

JL: careful in-depth study of the value iteration algorithm to identify why it is an approximation, important technical revisions (e.g., the intuition behind the confidence level, the reason for the approximate algorithm, clearly specifying the history-dependent policy space), proof of Lemma 1, statement of Lemma 3, expertise of Chow 2015 paper; reviewed drafts of the paper and proofs.

AT: explained Chow 2015 paper to MC during numerous discussions; reviewed drafts of the paper and proofs.

DL: integration of the parameter $m$ into the log-sum-exp approximation in proof of Theorem 1.

KS: expertise in PCSWMM, used PCSWMM to generate realistic simulations of the stormwater catchment in response to the design storm; interpreted the meaning of risk-sensitive safe sets in the context of the stormwater design application, which is a new application for reachability theory; reviewed drafts of the paper and proofs.

VC: generated the probability distribution for the surface runoff using empirical moments.

JF: provided initial golden nugget of ``can we formulate a risk-sensitive version of reachability-based control?"; worked with MC on numerous occasions to study risk measures and measure theory; reviewed drafts of the paper and proofs.

SJ: MC's research mentor and PI at SRI international over the summer when this paper was developed, through numerous discussions helped M.C. convey mathematical ideas more clearly, improved the proof of Lemma 2; reviewed drafts of the paper and proofs.

MP: JL's adviser, made numerous important revisions to the framing of the paper to take into account the approximate value-iteration method; reviewed drafts of the paper and proofs.

CT: MC's adviser, numerous discussions with M.C. during the development of the paper to improve accessibility of the ideas and techniques; reviewed drafts of the paper and proofs.

\addtolength{\textheight}{-2cm}   

\bibliography{references}

\end{document}